\documentclass{cccg14}
\usepackage{graphicx,amssymb,amsmath}

\title{A fast 25/6-approximation for the minimum unit disk cover problem}

\author{Paul Liu
\and Daniel Lu\thanks{Department of Computer Science, University of British Columbia. {\tt \{paul.liu.ubc, daniel.lawrence.lu\}@gmail.com}}}

\index{Liu, Lu}
\begin{document}
\thispagestyle{empty}
\maketitle

\begin{abstract}
Given a point set $P$ in $\mathbb{R}^2$, the problem of finding the smallest set of unit disks that cover all of $P$ is NP-hard. We present a simple algorithm for this problem with an approximation factor of $25/6$ in the Euclidean norm and $2$ in the max norm, by restricting the disk centers to lie on parallel lines. The run time and space of this algorithm is $O(n \log n)$ and $O(n)$ respectively. This algorithm extends to any $L_p$ norm and is asymptotically faster than known alternative approximation algorithms for the same approximation factor.

\end{abstract}

\section{Introduction}
Given a point set $P$ in $\mathbb{R}^2$, the \textit{unit disk cover problem} (UDC) seeks to find the smallest set of unit disks that cover all of $P$. This problem arises in applications to facility location, motion planning, and image processing \cite{fowler, hochbaummaass}.

In both the $L_2$ and $L_\infty$ norm, UDC is NP-hard \cite{fowler}. A \textit{shifting strategy} admits various polynomial time approximation algorithms in $d$ dimensions --- for some arbitrarily large integer shifting parameter $\ell$, it is possible to approximate to within $\left(1+\frac{1}{\ell}\right)^{d-1}$ \cite{hochbaummaass, gonzalez}. Since these algorithms rely on optimally solving the problem in an $\ell\times \ell$ square through exhaustive enumeration, they tend to have a slow time complexity that scales exponentially with $\ell$, making them impractical for large data sets. At the cost of incurring a constant approximation factor, the speed of the algorithm may be improved by constraining the disk centers to a unit square grid within the $\ell\times \ell$ square \cite{caltech1, caltech2}. 

If the disk centers are constrained to an arbitrary finite set of points, UDC becomes the discrete unit disk covering problem (DUDC), which is also NP-hard. However, DUDC has a number of different approximation algorithms, with the current state-of-the-art achieving a constant factor of 15 \cite{dudc}. 

In this paper, we present an algorithm that approximates UDC in the plane with the Euclidean and max norms by constraining the disk centers to a set of parallel lines. This algorithm is useable in practical settings and simple to implement. We show that, in the max norm, choosing a set of parallel lines distance $2$ apart achieves an approximation factor of 2. In the Euclidean norm, choosing a set of parallel lines distance $\sqrt{3}$ apart achieves an approximation factor of $25/6$. In both norms, the most costly step is simply from sorting the points. Consequently, the run time and space of the algorithm is $O(n \log n)$ and $O(n)$ respectively.

\renewcommand{\arraystretch}{1.5}
\begin{table}[htb]

\centering
\begin{tabular}{c|c|c|c}
\hline 
Paper & Approximation & Running Time & Year\tabularnewline
\hline 
\hline
%Hochbaum and Maass 
\cite{hochbaummaass}  & $\left(1+\frac{1}{\ell}\right)^2$ & $O\left(\ell^4 (2n)^{4\ell^2+1}\right)$ & 1985\tabularnewline
\hline 
%Feder and Greene 
\cite{gonzalez}  & $\left(1+\frac{1}{\ell}\right)$ & $O\left(\ell^{2}n^{6\ell\sqrt{2}+1}\right)$ & 1991\tabularnewline
\hline 
%Feder and Greene 
\cite{gonzalez}  & 8 & $O\left(n+n\log H\right)$ & 1991\tabularnewline
\hline 
%Br\"onnimann and Goodrich 
\cite{bronnimann} & $O(1)$ & $O(n^{3}\log n)$ & 1995\tabularnewline
\hline 
%Franceschetti, Cook, and Bruck 
\cite{caltech2}  & $\alpha\left(1+\frac{1}{\ell}\right)^{2}$ & $O(Kn)$ & 2001\tabularnewline
\hline 
%This paper 
Ours & $25/6$ & $O(n\log n)$ & 2014\tabularnewline
\hline 
\end{tabular}

\caption{A history of approximation algorithms for the unit disk cover problem in $L_2$. $n$ is the number of points in $P$. The shifting parameter $\ell$ is a positive integer which may be arbitrarily large. $H$ is the number of circles in the optimal solution. $\alpha$ is a constant between 3 and 6. $K$ is a factor at least quadratic in $l$ and polynomial in the size of the approximation lattice.}

\end{table}

\section{Line restricted unit disk cover}
\label{lrudc}
Here we explore a restricted variant of UDC:

Given a point set $P$ in $\mathbb{R}^2$, the \textit{line restricted unit disk cover} problem (LRUDC) seeks to find the smallest set of unit disks --- each with \textit{centers on a given set of parallel lines $S$} --- that cover all of $P$.

For certain carefully chosen sets of lines, LRUDC can be solved efficiently using greedy methods. However, with no restrictions on the placement and number of parallel lines in $S$, LRUDC is NP-hard by reduction from UDC.

\begin{theorem}
Using $O(n^2)$ parallel lines, UDC reduces to LRUDC.
\end{theorem}

\begin{proof}
Consider a circle arrangement $\mathcal A$ consisting of unit radius circles centered at each of the points in the point set $P$. For any circle $C$ in the optimal solution of UDC, let $F$ be the face in $\mathcal A$ in which the center of $C$ resides. Observe that moving this center to any point in $F$ does not change the subset of points in $P$ that $C$ covers. If the set $S$ of parallel lines intersects all faces in $\mathcal A$, then the optimal line-restricted solution can have disks centered in the same set of faces as in the unrestricted case. Hence, any optimal solution of LRUDC for this set of lines is an optimal solution of UDC. Since there are only $O(n^2)$ faces in $\mathcal A$, having one line for each of the faces suffices.
\end{proof}

As an aside, it is unknown whether LRUDC is NP-hard if only $O(n)$ parallel lines are used.

% ====================================================================================

\section{Approximation algorithms for UDC}
% =============================================

In our approximation algorithms for UDC, we use solutions to LRUDC on narrow vertical strips. The set $S$ of restriction lines we use for LRUDC will simply be uniformly spaced vertical lines. For each restriction line we will solve LRUDC confined to the subset of $P$ within a thin strip around the line. All points in $P$ will be in some strip, and we will choose the spacing between restriction lines so that a good approximation to UDC is obtained.

\subsection{A $2$-approximation with the $L_\infty$ norm}
\label{sec:max-norm-p3}
The max norm is a special case, as unit circles in the max norm are axis-aligned squares of width 2. We can take advantage of this fact to obtain a $2$-approximation algorithm.

\begin{enumerate}
\item Partition the plane into vertical strips of width $2$, and let the restriction line set $S$ be the set of vertical lines running down the centre of the strips.
\item For each non-empty strip, use the simple greedy procedure of inserting a square whose top edge is located at the topmost uncovered point. Repeat until all points in the strip are covered.
\end{enumerate}

The asymptotic cost of the algorithm is only $O(n\log n)$, as we need to sort the points by $x$-coordinate to partition them into the strips, and within each strip, we need to sort the points by $y$-coordinate to process the points in order of decreasing height.

The correctness of the greedy procedure in step 2 is easy to see, and is described in some detail by \cite{federgreene}. In fact, for this set of lines, this algorithm solves LRUDC optimally. This is because the greedy procedure is optimal for each strip, and the strips are all \textit{independent} from one another --- meaning that no point will be covered by squares from two different strips.

\begin{figure}
\centering
\includegraphics[width=7cm]{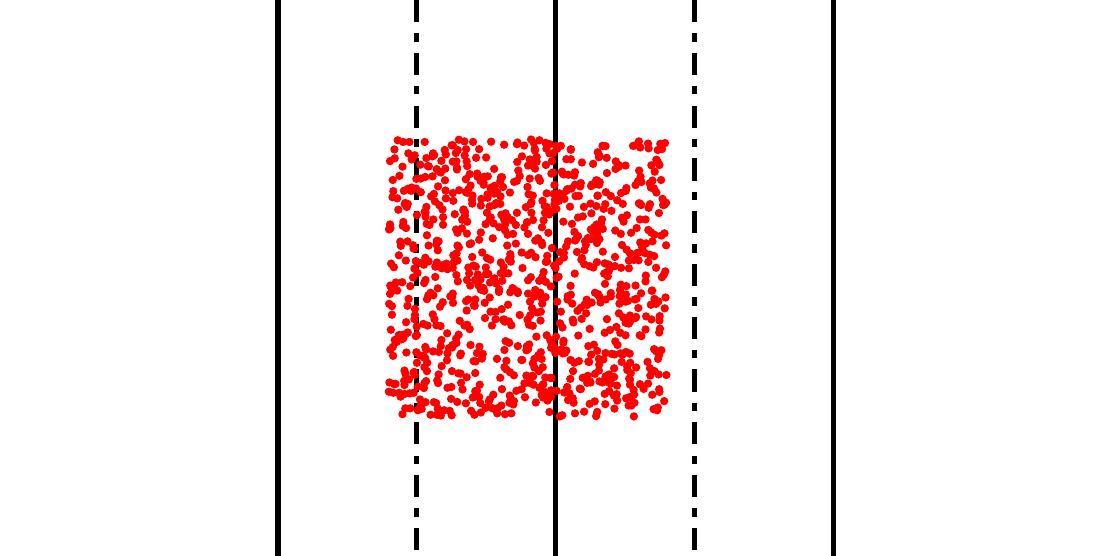}\\
$\,$ \\
\includegraphics[width=7cm]{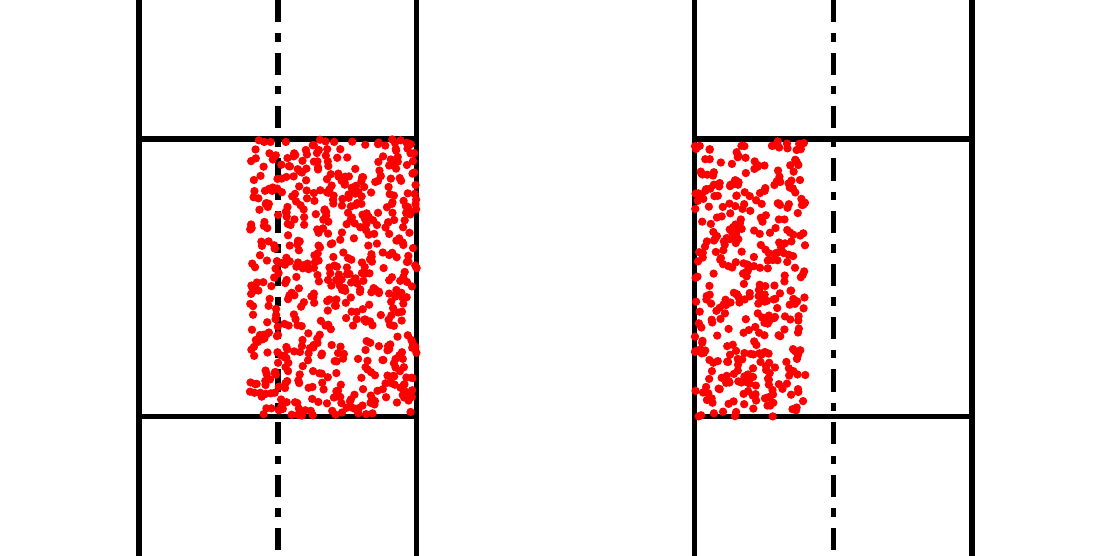}
\caption{Theorem \ref{square}. Top: a point set optimally coverable by one square. Bottom: a covering solution for each strip. Dashed lines denote restriction lines in $S$.}
\label{fig3}
\end{figure}

\begin{theorem}\label{square}
This algorithm is a 2-approximation for UDC in $L_\infty$.
\end{theorem}
\begin{proof}

For convenience, we define an \textit{$S$-restricted} solution to be any line-restricted solution covering all the points of $P$ using the same set $S$ of lines as our algorithm, but not necessarily the same set of circles as the one produced by our algorithm.

Let \textsc{opt} be an optimal solution for UDC in $L_\infty$. Each square in the optimal solution will intersect at most two strips, since each strip has the same width as the squares. We can construct an $S$-restricted solution, by simply using two $S$-restricted squares to cover each square in \textsc{opt} (see Figure \ref{fig3}). This uses exactly twice as many squares as \textsc{opt}. Since our algorithm solves LRUDC on $S$ optimally, it will be at least as good as the 2-approximation on each strip. As each strip is independent, our algorithm will be as good as the $S$-restricted solution over all strips.
\end{proof}

%===============================================

\subsection{A $5$-approximation with the $L_2$ norm}
\label{l2p2approx}

First, we present a simple $5$-approximation algorithm that forms the basis of our $25/6$-approximation algorithm.

\begin{enumerate}
\item Partition the plane into vertical strips of width $\sqrt{3}$. As before, let the restriction line set $S$ be the set of center lines of the strips.
\item For each non-empty strip, use the simple greedy procedure of inserting a circle positioned as low as possible while still covering the topmost uncovered point. Assume that all points in the strip are uncovered initially, and repeat until all points in the strip are covered.
\end{enumerate}

Note that the only difference between the algorithm above and the one for $L_\infty$ is the width of the vertical strip. With a width of $\sqrt{3}$, circles centred on a particular strip can cover points of neighbouring strips. Hence the strips are no longer independent, and we run the greedy procedure in step 2 assuming that all the points in the strip are uncovered initially (even though they may be covered by circles from different strips). Alternatively, we could remove points already covered by neighbouring strips as we go, but this makes no difference to the approximation factor or the asymptotic run time of the algorithm.

As before, parititioning the points into each strip is $O(n \log n)$ time. Within each strip, the subprocedure of greedily covering the circles can be done in $O(n_s \log n_s)$ time, where $n_s$ is the number of points in the strip. This is achieved by transforming the point covering problem into a segment covering problem instead. 

The reduction is as follows: from each point $p$ in the strip draw a unit circle $C_p$ centred at $p$. The circle $C_p$ intersects the restriction line of the strip in two points, creating a segment $s_p$ between the two points. If the centre of an $S$-restricted circle is placed anywhere on $s_p$, it will cover $p$. Hence to cover all the points in the strip, we simply have to stab all the segments $\{s_p\}_{p\in P}$ with points representing centres of $S$-restricted circles. The strategy of greedily covering the topmost point reduces to choosing the stabbing point as low as possible, while still stabbing the topmost unstabbed segment. This can be done in $O(n_s \log n_s)$ time via sorting the segments by $y$-coordinate.

The correctness of the greedy subprocedure in step 2 follows from the same logic as Section \ref{sec:max-norm-p3}.

The argument that our algorithm is a 5-approximation is based on the following fact: for each circle $C$ in an optimal solution \textsc{opt} of UDC, there exists an $S$-restricted solution which covers each $C$ entirely using at most five circles. Furthermore, this solution is redundant in that the points of each strip are covered completely by $S$-restricted circles on that strip. We call such a solution \textit{oblivious}, as it does not take into account points covered by circles of neighbouring strips. Note that our algorithm produces a solution that is at least as good as any oblivious $S$-restricted solution, as each strip is solved optimally by our algorithm. It follows that our algorithm is also a 5-approximation.

It is necessary in the worst case to cover $C$ entirely since an adversary may provide an input point set $P$ consisting of arbitrarily many points coverable by a single circle (Figures \ref{fig1} and \ref{fig2}). The following proofs use a straightforward application of geometry to establish bounds on the number of $S$-restricted circles to cover $C$. There are two possible cases --- either $C$ intersects two strips, or $C$ intersects three strips.

\begin{obs}
\label{4c}
Let \textsc{opt} be the set of optimal circles for UDC. Suppose that the centre of a circle $C\in \textsc{opt}$ does not lie within $1-\frac{\sqrt{3}}{2}$ of a restriction line. Then, any oblivious $S$-restricted solution will require at least four circles to cover $C$.

Moreover, there exists an oblivious $S$-restricted solution which uses exactly four circles to cover $C$.
\end{obs}

\begin{proof}
Without loss of generality, let $C$ be centered at $(x_c,0)$ where $1-\frac{\sqrt{3}}{2}\leq x_c \leq 3\frac{\sqrt{3}}{2}-1$. For $x_c$ in this range, $C$ intersects two strips. Let the corresponding restriction lines be called $\mathcal L_1$ and $\mathcal L_2$ and be placed at $x=0$ and $x=\sqrt{3}$ respectively. Consider the strip boundary, a vertical line $\mathcal L_{12}$ at $x=\frac{\sqrt{3}}{2}$. The intersection of $C$ with this line forms a segment of length greater than 1 but smaller than 2. To cover $C$ entirely, this segment must be covered. For both strips that $C$ intersects, the algorithm would cover this segment, as each strip is oblivious that the neighbouring strip may have covered the same segment. Since each line-restricted circle can only cover a segment of length 1 on $\mathcal L_{12}$, each strip would need two circles, resulting in a total of 4.

It is easy to see that $C$ is covered by the four circles centred at $\left(0,\frac{1}{2}\right)$, $\left(0,-\frac{1}{2}\right)$, $\left(\sqrt{3},\frac{1}{2}\right)$, $\left(\sqrt{3},-\frac{1}{2}\right)$ (see Figure \ref{fig1}). Note that the obliviousness constraint is satisfied as the circles within each strip do not depend on circles from neighbouring strips to cover the points from $C$.
\end{proof}

\begin{figure}[h]
\centering
$\mathcal L_1$ \hspace{5pt} $\mathcal L_{12}$ \hspace{5pt} $\mathcal L_2$\\
$\,$ \\
\includegraphics[width=7.5cm]{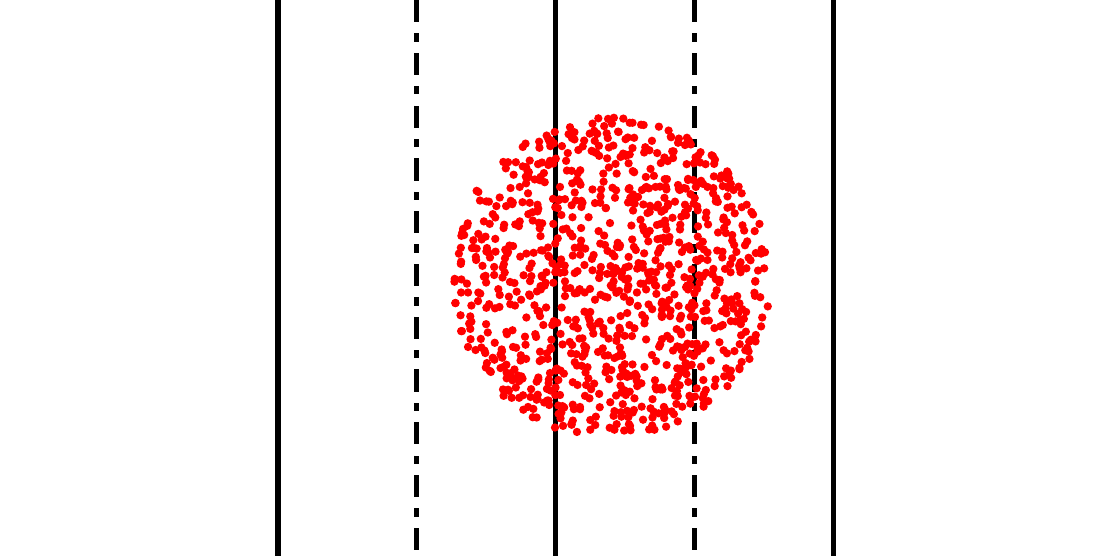}\\
$\,$ \\
\includegraphics[width=7.5cm]{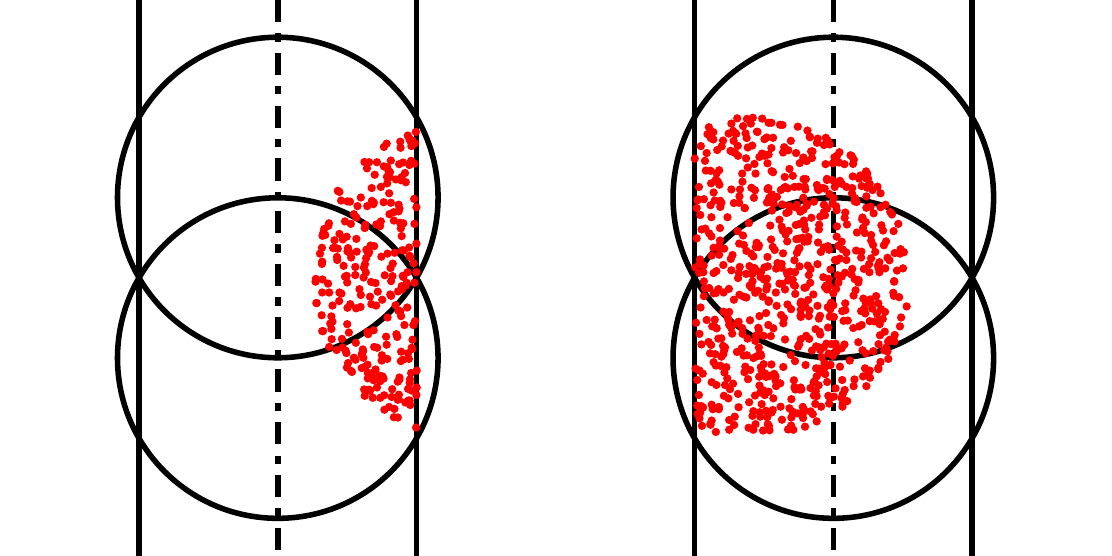}
\caption{Observation \ref{4c}. Top: a set of points optimally coverable by one circle. Bottom: a covering solution for each strip. Dashed lines denote restriction lines in $S$.}
\label{fig1}
\end{figure}

\begin{obs}
\label{5c}
Let \textsc{opt} be the set of optimal circles for UDC. Suppose that the centre of a circle $C\in \textsc{opt}$ lies within $1-\frac{\sqrt{3}}{2}$ of a restriction line. Then, any oblivious $S$-restricted solution will require at least five circles to cover $C$.

Moreover, there exists an oblivious $S$-restricted solution which uses exactly five circles to cover $C$.
\end{obs}

\begin{proof}
Without loss of generality, let $C$ be centered at $(x_c,0)$ where $0\leq x_c < 1-\frac{\sqrt{3}}{2}$. For $x_c$ in this range, $C$ intersects three strips. Let the corresponding restriction lines be called $\mathcal L_1$, $\mathcal L_2$ and $\mathcal L_3$ and be placed at $x=-\sqrt{3}$, $x=0$ and $x=\sqrt{3}$ respectively. Consider the two strip boundaries, vertical lines $\mathcal L_{12}$ at $x=-\frac{\sqrt{3}}{2}$ and $\mathcal L_{23}$ at $x=\frac{\sqrt{3}}{2}$. The intersection of $C$ with $\mathcal L_{23}$ forms a segment centered at $y=0$ of length greater than 1 but smaller than 2, and the intersection of $C$ with $\mathcal L_{12}$ forms a segment centered at $y=0$ of length smaller than 1. To cover $C$ entirely, these segments must be covered. Since each line-restricted circle can only cover a segment of length 1 on the strip boundary, and each strip is oblivious that the neighbouring strip may have covered the same segment, strips 2 and 3 would need two circles each and strip 1 would need one circle, resulting in a total of 5.

Finally, it is easy to see that $C$ is covered by the five circles centred at $\left(-\sqrt{3},0\right)$, $\left(0,\frac{1}{2}\right)$, $\left(0,-\frac{1}{2}\right)$, $\left(\sqrt{3},\frac{1}{2}\right)$, $\left(\sqrt{3},-\frac{1}{2}\right)$ (see Figure \ref{fig2}).
\end{proof}

\begin{figure}[h]
\centering
\hspace{12pt} $\mathcal L_1$ \hspace{34pt} $\mathcal L_2$ \hspace{34pt} $\mathcal L_3$\\
$\,$ \\
\includegraphics[width=9cm]{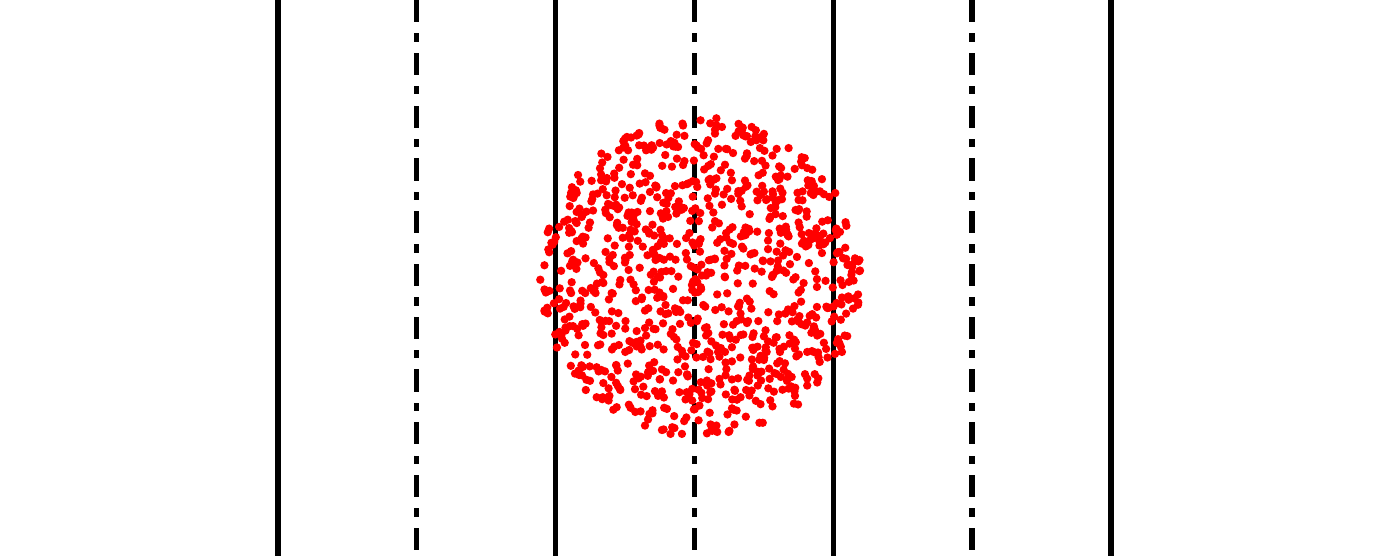}\\
$\,$ \\
\includegraphics[width=9cm]{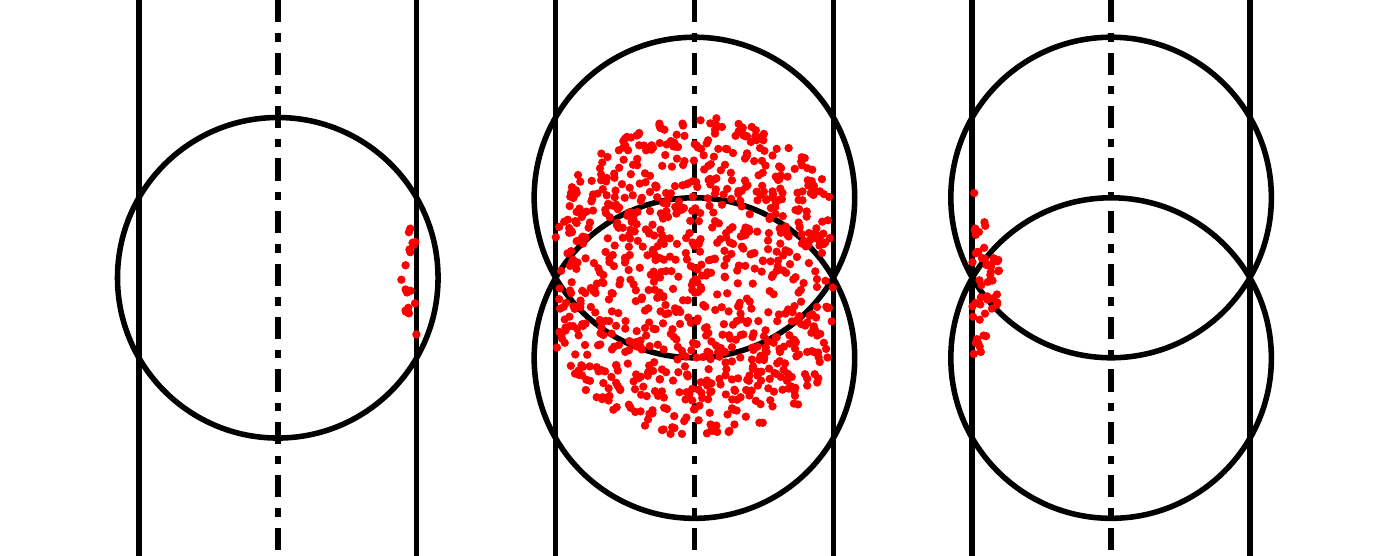}
\caption{Observation \ref{5c}. Top: a set of points optimally coverable by one circle. Bottom: a covering solution for each strip. Dashed lines denote restriction lines in $S$.}
\label{fig2}
\end{figure}

\begin{theorem}
\label{5app}
This algorithm is a $5$-approximation for UDC in $L_2$.
\end{theorem}

\begin{proof}
Since our algorithm solves each strip optimally, it produces a solution that is at least as good as any oblivious $S$-restricted solution. We have shown that there exists such a solution with an approximation factor of 5, which follows directly from Observations \ref{4c} and \ref{5c}. Hence, our algorithm has an approximation factor of at most 5.
\end{proof}

%===============================================
\subsection{Improving the 5-approximation to a 25/6-approximation}
\label{improve}

To improve the 5-approximation algorithm to a $25/6$-approximation algorithm, we employ a ``smoothing'' technique. From the calculations in Observation \ref{5c}, the region where a circle $C$ in an optimal solution requires five circles to cover is only $2-\sqrt{3}$ wide. In all other cases, $C$ can be covered by only four circles. Since $2-\sqrt{3}$ is less than one-sixth of the width of the entire strip $\sqrt{3}$, it is intuitive that we can do better than a 5-approximation. Here, we show that by shifting the strip partition, we can smooth out the regions that require 5-circles and achieve a $25/6$-approximation.

To be precise, we define a strip partition with shift $\alpha$ to be the partition of $\mathbb{R}^2$ into width $\sqrt{3}$ vertical strips, where the boundaries of the strips are located at $x=\alpha + k\sqrt{3}$, $k\in \mathbb{Z}$. As usual, our restriction lines are in the centres of these strips. Our algorithm with the smoothing technique is:

\begin{enumerate}
\item For $\alpha=0,\frac{\sqrt{3}}{6},\frac{2\sqrt{3}}{6},\ldots, \frac{5\sqrt{3}}{6}$, partition the plane into vertical strips of width $\sqrt{3}$ with shift $\alpha$ and use the 5-approximation algorithm.
\item Return the best of the six solutions obtained above.
\end{enumerate}

\begin{theorem}
The algorithm with smoothing approximates is a $25/6$-approximation for UDC in $L_2$.
\end{theorem}

\begin{proof}
Let \textsc{opt} be the set of optimal circles for UDC, and let $S_1,\ldots,S_6$ be the six sets of shifted line sets used in our algorithm. For each of the six shifts, there exists an oblivious $S_i$-restricted solution. 

Suppose that for $i=1,\ldots,6$, there are $q_i$ circles in \textsc{opt} with centers $(x,y)$ satisfying 
\begin{align}
\alpha_i + k\sqrt{3} + \frac{5}{12}\sqrt{3} \leq x\leq \alpha_i + k\sqrt{3} + \frac{7}{12}\sqrt{3}\label{eq1}
\end{align}
for $k\in\mathbb{Z}$, where $\alpha_i = (i-1)\frac{\sqrt{3}}{6}$. Note that since these six ranges fill the plane, $\sum q_i=|\textsc{opt}|$. 

According to Observations \ref{4c} and \ref{5c}, each solution has five circles for every circle in \textsc{opt} that has center $(x,y)$ satisfying
\begin{align}
\alpha_i + k\sqrt{3} + \sqrt{3}-1 \leq x \leq \alpha_i + k\sqrt{3} + 1\label{eq2}
\end{align}
and four circles for every other circle in \textsc{opt}. Since the range in Equation \ref{eq2} is a subrange of that in Equation \ref{eq1}, it follows that an oblivious $S_i$-restricted solution of Section \ref{l2p2approx} uses no more than
\begin{align}
5q_i + 4\sum_{\substack{j=1\\ j \neq i}}^6 q_j
\end{align}
circles. 

For $i=1,\ldots,6$, let $A_i$ be the $i$-th candidate solution generated by our algorithm and let $A^*$ be the solution with fewest circles out of the 6 $A_i$'s. Since each $A_i$ is at least as good as any oblivious $S_i$-restricted solution, we have the inequality:
\begin{align}
|A^*| & = \min_{i=1,..,6} |A_i| \\
&\leq \min_{i=1,..,6} \left[5q_i + 4\sum_{\substack{j=1\\ j \neq i}}^6 q_j\right]\\
&=4|\textsc{opt}| + \min_{i=1,..,6} q_i\\
&\leq 4|\textsc{opt}| + \frac{1}{6}|\textsc{opt}| = \frac{25}{6}|\textsc{opt}|
\end{align}
Hence the output of our algorithm is a $25/6$ approximation to the unit disk cover problem.
\end{proof}

% ====================================================================================

\section{Extensions}

The approximation algorithm outlined above can be applied to any $L_p$ norm --- one simply has to figure out the worst case number of oblivious line-restricted circles to cover an arbitrary circle in the plane. For each norm, the optimal line spacing varies. However, our algorithm is guaranteed to produce constant factor approximations when the spacing less than $2$.

Our algorithm can also be extended to higher dimensions. A natural extension is to use a collection of uniformly spaced parallel lines, and solve LRUDC in a small tube surrounding each line. In this way, we can obtain approximations in arbitrarily large $d$ dimensions, albeit with an approximation factor that scales exponentially with $d$. In particular, applying this technique to the $L_\infty$ norm gives a $2^{d-1}$-approximation in $d$ dimensions, matching an earlier result by \cite{gonzalez}.

Finally, our algorithm applies to covering objects more general than points, such as polygonal shapes. Our proofs only rely on the fact that any optimal circle can be covered entirely with a constant number of oblivious line-restricted circles. The fact that we are covering points is not used.
% ==================================================================================

\section{Concluding Remarks}

We presented a simple algorithm to approximate the unit disk cover problem within factor of $25/6$ in $L_2$ and within a factor of 2 in $L_\infty$.  

The algorithm runs in $O(n \log n)$ time $O(n)$ space, with the most time consuming step being a simple sorting of the input. On a practical level, we believe our algorithm has a good mix of performance and simplicity, with a typical implementation of no more than 30 lines of C++.

We wonder what the best approximation an oblivious line-restricted approach can achieve for the unit disk cover problem. For the $L_\infty$ norm, we saw that 2 was the best possible approximation factor for lines spaced equally apart. Similarly, one can show for the $L_2$ norm that a lower bound of $15/4$ is the best that can be done with an oblivious algorithm for equally spaced lines. It would be interesting to see if these lower bounds can be broken by an oblivious algorithm once the equal spacing condition is removed. Finally, an analysis of the optimal spacing in other $L_p$ norms would be interesting as well.

\section{Acknowledgements}
The authors would like to thank Professors David Kirkpatrick and Will Evans for their limitless patience and guidance, as well as Kristina Nelson for reading the early drafts.

\small
\bibliographystyle{abbrv}

%\newpage
%\section*{Appendix}
%Submissions should not exceed six pages. Authors who feel that additional details are necessary should include a clearly marked appendix, which will be read at the discretion of the Program Committee. 

%Proceedings will be published online at the conference website. There is no page limit for papers published in the proceedings. 

\end{document}